\documentclass[12pt]{article}
\usepackage{amsmath}
\usepackage{amsthm}
\usepackage{amssymb}
\usepackage{amsfonts}
\usepackage{plain}
\usepackage{graphicx}
\usepackage{tikz}
\usepackage[normalem]{ulem}
\usepackage[utf8]{inputenc}
\usepackage{xcolor}
\newtheorem{theorem}{Theorem}[section]

\newtheorem{definition}{Definition}
\theoremstyle{definition}

\theoremstyle{remark}

\usetikzlibrary{trees,arrows,positioning,fit,calc}
\tikzset{block/.style={draw, thick, text width=2cm, minimum height=0.75cm, align=center}
		}

\begin{document}

\title{Multi-signer Strong Designated Multi-verifier Signature Schemes based on Multiple Cryptographic Algorithms}
\author{Neha Arora, R. K. Sharma}
\date{}
\maketitle
\begin{center}
\textit{Department of Mathematics, Indian Institute of Technology, New Delhi, 110016, India}
\end{center}

\begin{abstract}
A designated verifier signature scheme allows a signer to generate a signature that only the designated verifier can verify. This paper proposes multi-signer strong designated multi-verifier signature schemes based on multiple cryptographic algorithms and has proven their security in the random oracle model.
\end{abstract}
2010 Math. Sub. Classification: 94P60, 94A62, 68P25, 68P30
\footnote{emails: nehaarora1907@gmail.com (Neha), rksharmaiitd@gmail.com (Rajendra)} 

\section{Introduction}

\textit{Digital signature} \cite{diffie1976new,elgamal1985public,rivest1978method} is a mathematical algorithm used to validate the authenticity, integrity, and non-repudiation of a message or document. In digital signature algorithms, anyone having the signature and signer's public key can validate the signature, which is not desired in some cases. Therefore, the undeniable signature \cite{chaum1989undeniable} was introduced by David Chaum and Hans van Antwerpen in 1989. 
An \textit{undeniable signature} is a digital signature that requires the signer's cooperation to verify signatures.
However, the validity of an undeniable signature can be ascertained by anyone issuing a challenge to the signer and testing the signer's response.
To solve this problem, Jakobsson Markus, Kazue Sako, and Russell Impagliazzo introduced a \textit{designated verifier signature scheme} \cite{jakobsson1996designated} in 1996.
It provides authentication of a message, and the signer chooses a designated verifier in advance, which makes it non-interactive. However, it does not provide non-repudiation property. Also, in this scheme, the verifier can generate a transcript and convince the third party that the signer created the signature.  
In the same year, a stronger version is introduced as a \textit{strong designated verifier signature scheme}. However, in 2003, Guilin Wang designed an attack in \cite{wang2003attack} on the Jakobsson scheme. Later, in 2003, an efficient strong designated verifier signature scheme  \cite{saeednia2003efficient} was introduced by Shahrokh Saeednia, Steve Kremer, and Olivier Markowitch.
\\ However, in 2007, Ji-Seon Lee and Jik Hyun Chang in \cite{lee2007strong} found that Saeednia scheme \cite{saeednia2003efficient} was not secure, and it would reveal the signer's identity if the secret key of the signer is compromised. Later, they provide an improved version \cite{lee2009comment}. However, in 2012, Liao Da-jian and Tang Yuan-sheng found that \cite{lee2009comment} does not protect the identity of the signer, and it can be revealed if compromised. Thus, they designed an improved version \cite{da2012comment} of \cite{lee2009comment}  which protects the signer's identity but lost the security properties of the designated verifier signature scheme. Also, in 2017, \cite{khan2017secure} proved that the scheme \cite{lee2009comment} is not secure, and signature can be forged without knowing the signer's private key and designed a new strong designated verifier signature scheme.
\\ In 2015, Pankaj Sarde and Amitabh Banerjee reviewed \cite{saeednia2003efficient} and proposed strong designated verifier signature scheme based on discrete logarithm problem \cite{sarde2015strong}.
\\ In 2018 and later in 2021, Nadiah Shapuan and Eddie Shahril Ismail proposed a strong designated verifier signature scheme with two hard problems : Factoring problem and discrete logarithm problem \cite{shapuan2018new,shapuan2021strong}.
\\ Simultaneously, Few multi signer designated verifier signature scheme \cite{islam2012certificateless,li2005id,zhang2008multi}, designated multi verifier signature schemes \cite{laguillaumie2004multi,laguillaumie2007multi,ming2008universal,yang2008strong} and multi-signer designated multi-verifier signature schemes \cite{chang2011id,deng2018id,tzeng2004nonrepudiable} were introduced.
\\ In the case of Multi-signer designated multi-verifier signature schemes, signatures can be generated in two different ways :
\begin{enumerate}
\item All signers come together and cooperate to generate a new signature.
\item All signers generate their signature and make it public to the system, and then, the system generates a new signature using all available signatures.
\end{enumerate}
Multi-signer designated multi-verifier signature schemes has application in various areas. A few of them are mentioned below:
\begin{enumerate}
\item Let A be a group of board of Directors of a company, and they have signed some crucial and confidential data which may lead to the increase in profits and price of shares of a company and make its encoded copy public such that each shareholder can verify if the decisions will benefit them or not.

\item \textbf{Application using Blockchain Technology:} 
\\ To apply our scheme to the blockchain, we made a few assumptions :
\begin{enumerate}
\item Our example best suits a small blockchain where participants are limited, and the transaction information is shared among them only.
\item Our platform is a smart contract-enabled permissioned Blockchain Network.
\item Each signer and verifier signed a smart contract, and any cheating and violation led to the blocked account.
\item Only designated verifiers are allowed to verify and validate the information.
\item For every transaction information shared, the participants will be the signers, and the remaining will act as verifies.
\item Moreover, ownership is not assigned to a single person. Thus, overriding, editing, or deleting the transaction or any other related information is not possible without the consent of a fixed number of participants.
\item A new block is generated after a fixed predetermined time interval.
\item The hash value of all transactions in a fixed interval are merged as a root hash in Merkle tree.
\item Each block is linked with the previous block.
\item Each block contains block version, Merkle tree root hash, timestamp, nBits, nounce, parent block hash.
\begin{center}
\begin{small}
\begin{tikzpicture}
\begin{scope}

  \node[block] (a) {Block version};
  \node[block,right=of a] (b) {{Merkle tree root hash}};
  \node[block,right=of b] (c) {{Parent block hash}};
   \node[block] (d) at ([yshift=-1cm]$(b)!1!(a)$) {Timestamp};
  \node[block,right=of d] (n) {nBits};
  \node[block,right=of n] (v) {Nounce};
  \node[block] (e) at ([yshift=-3cm]$(b)!1!(a)$) {TX};
  \node[block,right=of e] (g) {TX};
  \node[block,right=of g] (h) {TX};
  \node[draw,inner xsep=5mm,inner ysep=10mm,fit=(a)(b) (c) (d) (n) (v) (g) (h),label={270:\textbf{{\large Block Structure}}}]{};
  \node[draw,inner xsep=2mm,inner ysep=2mm,fit=(a) (b) (c) (d) (n) (v),label={90: \textbf{{\large Block Header}}}]{\vspace{2 cm}};
  \node[draw,inner xsep=2mm,inner ysep=2mm,fit=(e) (g) (h),label={90: \textbf{{\Large Transaction Counter}}}]{};

\end{scope}
\end{tikzpicture}
\end{small}
\end{center}
\end{enumerate}
We shall show it with an example:
\begin{enumerate}
\item We assume we have 20 participants in our Blockchain Network and each member is given an identity $P_i$ , $1 \leq i \leq 20$.
\item A new block is generated after every $t'$ minutes.
\item Let $T_{i,j}$ be the transaction id where $i$ represents the block number and $j$ represents the transaction number of $i$-th block. 
\item Let $P_1$ transfers some amount to $P_2, P_3, P_5, P_7$ and $P_9$ transfer some amount to $P_1, P_3, P_5,P_{11}$ in a fixed time interval $[0, t'] $. Let $S_1 = \{P_1, P_2, P_3, P_5, P_7, P_9,B_{11}\}$ be the set of signers. None of them wants to reveal the exact amount, but their ids, and encoded value of the change in their respective balance. In this case, each signer creates and sends the signature to the system. The system generates a new signature and broadcasts it to all remaining participants, and each verifier, within a specific time, can verify and validate the transaction.

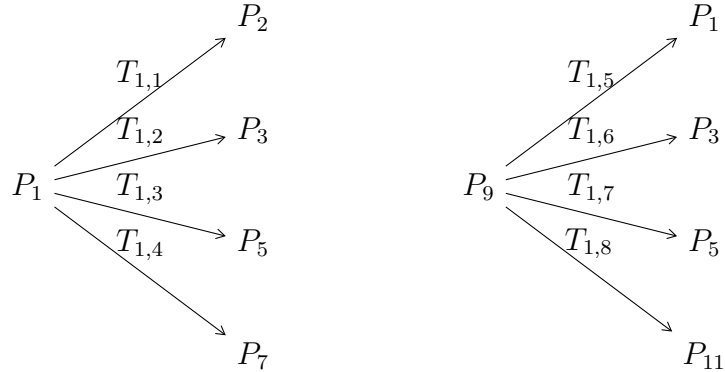
\tikzstyle{level 1}=[level distance=30mm, sibling distance=15mm]
\begin{center}
\begin{tikzpicture}[grow=right,->,>=angle 60]
\begin{scope}[xshift=0]
  \node {$P_1$}
    child {node {$P_7$}  edge from parent node[above]
                {$T_{1,4}$}}
          child {node {$P_5$} edge from parent node[above]
                {$T_{1,3}$}}
          child {node {$P_3$} edge from parent node[above]
                {$T_{1,2}$}}
    child {node {$P_2$}  edge from parent node[above]
                {$T_{1,1}$}};

\end{scope}

\begin{scope}[xshift=6cm]
  \node {$P_9$}
    child {node {$P_{11}$} edge from parent node[above]
                {$T_{1,8}$}}
          child {node {$P_5$}  edge from parent node[above]
                {$T_{1,7}$}}
          child {node {$P_3$}  edge from parent node[above]
                {$T_{1,6}$}}
    child {node {$P_1$}  edge from parent node[above]
                {$T_{1,5}$}
      };

\end{scope}

\end{tikzpicture}
\end{center}
\item Let $P_7$ transfer some amount to $P_{13},P_{15}$ , $P_4$ to $P_2, P_7, P_{13}$ , and $P_3$ to $P_4, P_5, P_{11},P_{13}$ in a fixed time interval $[t', 2t']$. Then, $S _2= \{P_2, P_3, P_4, P_5, P_7, P_{11}, P_{13}, P_{15}\}$ be the new set of signers and remaining participants will act as  verifiers and they can verify and validate the transactions with a specific time.

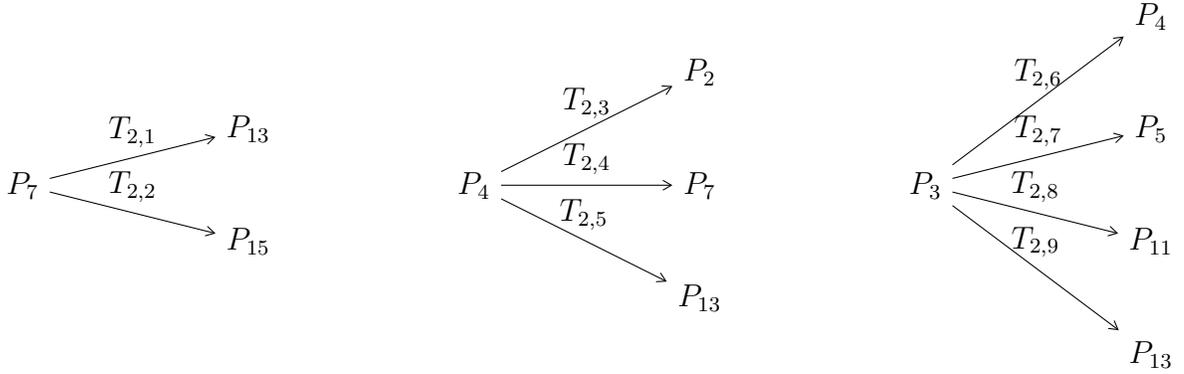
\tikzstyle{level 1}=[level distance=30mm, sibling distance=15mm]
\begin{tikzpicture}[grow=right,->,>=angle 60]
\begin{scope}[xshift=0]
  \node {$P_7$}
    child {node {$P_{15}$}  edge from parent node[above]
                {$T_{2,2}$}}
          child {node {$P_{13}$} edge from parent node[above]
                {$T_{2,1}$}};     
\end{scope}

\begin{scope}[xshift=6cm]
  \node {$P_4$}
    child {node {$P_{13}$} edge from parent node[above]
                {$T_{2,5}$}}
          child {node {$P_7$}  edge from parent node[above]
                {$T_{2,4}$}}
          child {node {$P_2$}  edge from parent node[above]
                {$T_{2,3}$}};
\end{scope}

\begin{scope}[xshift=12cm]
  \node {$P_3$}
    child {node {$P_{13}$} edge from parent node[above]
                {$T_{2,9}$}}
          child {node {$P_{11}$}  edge from parent node[above]
                {$T_{2,8}$}}
          child {node {$P_5$}  edge from parent node[above]
                {$T_{2,7}$}}
    child {node {$P_4$}  edge from parent node[above]
                {$T_{2,6}$}
      };

\end{scope}
\end{tikzpicture}

\end{enumerate}
This way, we can record and verify multiple transactions compactly. Also, the privacy of each transaction can be maintained by keeping individual amounts secret and making hash value the total amount public. It also decreases the size of the blocks and reduces the storage cost, making it cheaper to operate.

\end{enumerate}

\vspace{0.5 cm}

So far, we have seen a few multi-signer designated verifier signature schemes, designated multi-verifier signature schemes, and multi-signer designated multi-verifier signature schemes based on either single hard problem or bilinear pairing. This paper proposes two individual multi-signer strong designated multi-verifier signature schemes and combines them to enhance their security. Our final scheme is based on multiple hard problems, and it requires the authorization of each participating signer. Also, our scheme is non-interactive as the system will act as an intermediator. Later, we showed that our scheme is a strong, unforgeable, non-transferable designated verifier scheme. 

This paper is organized as follows: In Section 2, we define factoring problem, discrete logarithm problem, bilinear pairing, strong designated verifier, strong designated verifier signature scheme, and multi-signer designated multi-verifier signature schemes. Section 3 is divided into four subsections. Section 3.1 is an extension of \cite{zhang2008multi} to multi-verifier scheme, and it is based on bilinear pairing. Section 3.2 is an extension of \cite{shapuan2018new,shapuan2021strong} to multi-signer multi-verifier scheme, and it is based on two hard problems: factoring problem and discrete logarithm problem. Section 3.3 is an improvisation of section 3.2 based on the elliptic curve. Section 3.4 gives a glimpse of the combination of the first and second algorithm. In section 4, we have done the security analysis of our scheme.


\section{Preliminaries}

\begin{definition}
\textbf{Factoring problem :}
Let $n$ be a large positive composite integer such that $n \ = \ pq$, where $p$ and $q$ are large primes. Then, solving to find the factorization of $n$ is called factoring problem.
\end{definition}

\begin{definition}
\textbf{Discrete Logarithm problem:} 
Let $p,q$ be primes such that $q$ divides $p-1$. Let G be a multiplicative group of order $p$, and $g \in G$ such that $|g| = q $ and $y \ \equiv \ g^x \ (mod \ n)$, then for given $g$ and $y$, computation of $x$ is called as discrete logarithm problem.
\end{definition}

\begin{definition}
\textbf{Bilinear pairing :}

Let $G_1$ and $G_2$ be two cyclic groups (w.r.t. addition) of prime order $p$ with generators $P_1$ and $P_2$, respectively. Define a map $e : G_1 \times G_2 \to G_{\mathcal{T}}$, where $G_{\mathcal{T}}$ is a group (w.r.t. multiplication) of order $p$. Then, the map $e$ is s.t.b. $\mathbb{F}_p$-bilinear if $e(a P_1, b P_2) \ = \ e(P_1,P_2)^{a b} \ \ \forall a,b \in \mathbb{F}_p $, non-degenerate if $e(P_1,P_2) \neq 1_{G_{\mathcal{T}}} $, and symmetric if $G_1 \ = \ G_2$.

The pairing $e$ is an admissible bilinear map if it is non-degenerate and it can be efficiently computable.


\end{definition}

\begin{definition} (\cite{saeednia2003efficient} )
\textbf{Strong designated Verifier :}
Let P(S,V) be a protocol for signer S to prove the truth of the statement $\omega$ to verifier V. We say that P(S,V) is strong designated verifier proof if anybody can produce identically distributed transcripts that are indistinguishable from those of P(S,V) for everybody, except for Verifier V.

\end{definition}

\begin{definition}
\textbf{Strong Designated Verifier Signature Scheme :}
A strong designated verifier signature scheme consists of the following algorithms:
\begin{enumerate}
\item Set up : $sp \ \leftarrow \ Setup(K)$, where $K$ is a security parameter, $sp$ is system parameters.
\item Key generation : $ (sk_U, pk_U) \ \leftarrow \ KeyGen(sp,K) $, where $sk_U$ and $pk_U$ are the secret key and public key of user $U$ respectively.
\item Signature Algorithm $ Sign_{{S} \ \rightarrow \ {D}}$ : $$ sign_S \ \leftarrow \ Sign_{{S} \ \rightarrow \ {D}} (M,sk_S,pk_D) $$
where $M \ \in \ \{0,1\}^* $ is the message to be signed, $sk_{S} $ denote the private key of signer $S$ and $pk_{D}$ denote the public key of designated verifier $D$.   
\item Verification Algorithm $ Verify_{D}$ :
$$ \{accept, \ reject \} \ \leftarrow \ Verify_{D} (M,sk_{D},pk_{S}, sign_S) $$ 
where $sk_{D}$ denote the secret key of designated verifier $D$ and $pk_{S}$ denote the public key of signer $S$.
\end{enumerate}

\end{definition}

\begin{definition}
\textbf{Multi-signer strong designated Multi-verifier signature scheme :}
A Multi-signer strong designated Multi-verifier signature scheme consists of the following algorithms:
\begin{enumerate}
\item Set up : $sp \ \leftarrow \ Setup(K)$, where $K$ is a security parameter, $sp$ is system parameters.
\item Key generation : $ (sk_U, pk_U) \ \leftarrow \ KeyGen(sp,K) $, where $sk_U$ and $pk_U$ are the secret key and public key of user $U$ respectively.
\item Signature Algorithm $Sign_{(S_1,S_2,...,S_n) \ \rightarrow \ (D_1,D_2,...,D_m)}$ : $$ sign_A \ \leftarrow \ Sign_{(S_1,S_2,...,S_n) \ \rightarrow \ (D_1,D_2,...,D_m)} (M,sk_{S_1},sk_{S_2},...,sk_{S_n},pk_{D_1},pk_{D_2},...,pk_{D_m}) $$
where $M \ \in \ \{0,1\}^* $ is the message to be signed, $(sk_{S_1},sk_{S_2},...,sk_{S_n}) $ denote the private key of $n$ signers $(S_1,S_2,...,S_n)$ and $(pk_{D_1},pk_{D_2},...,pk_{D_m})$ denote the public key of $m$ designated verifiers $(D_1,D_2,...,D_m)$.   
\item Verification Algorithm $ Verify_{(D_1,D_2,...,D_m)}$ :
$$ \{accept, \ reject \} \ \leftarrow \ Verify_{(D_1,D_2,...,D_m)} (M,sk_{D_1},sk_{D_2},...,sk_{D_m},pk_{S_1},pk_{S_2},...,pk_{S_n}, sign_A) $$ 
where $(sk_{D_1},sk_{D_2},...,sk_{D_m})$ denote the secret key of $m$ designated verifiers $(D_1,D_2,...,D_m)$ and $(pk_{S_1},pk_{S_2},...,pk_{S_n})$ denote the public key of $n$ signers $(S_1,S_2,...,S_n)$.
\end{enumerate}

\end{definition}

\section{Proposed Schemes}

This paper defines a secure, strong designated signature scheme with multi-signers and multi-verifiers based on multiple hard problems. Our scheme is divided into four subsections. Section 3.1 is an extension of \cite{zhang2008multi} to multi-verifier scheme, and it is based on bilinear pairing. Section 3.2 is an extension of \cite{shapuan2018new,shapuan2021strong} to multi-signer multi-verifier scheme, and it is based on two hard problems: factoring problem and discrete logarithm problem. Section 3.3 is an improvisation of section 3.2 based on the elliptic curve. We have illustrated the algorithm with examples. We can combine the first algorithms with the second algorithm or an improvised version of the second algorithm by making minor changes described in section 3.4.

Let $A = \{ S_1,S_2,...,S_n \} $ be the group of signers and $B = \{ D_1, D_2,...,D_m \} $ be the group of designated verifiers. Here $A$ and $B$ are systems that stores information which is made public by $S_i$ and $D_j$ respectively and then they interact with each other $l$ times, where $l$ is a predetermined number.

\subsection{Algorithm 1}


\vspace{0.3 cm}

\hspace{0.4 cm} \textbf{{\large Set up}}
  
\vspace{0.2 cm}  

Let $K$ be the security parameter and $sp \ = \ (G,G_\tau,g,p,e,H_1) \ \leftarrow \ Setup(K)$, where
\begin{itemize}
\item $G$ is an additive cyclic group of large prime order $p$ and $g$ is a generator of $G$,
\item $G_\tau $ is a multiplicative cyclic group of prime order $p$,
\item $e : \ G \times G \ \rightarrow \ G_\tau$ is symmetric admissible bilinear pairing map,
\item $H_1 : \{0,1\}^* \ \rightarrow \ G $ is a cryptographically secure hash function.
\end{itemize}

\vspace{0.3 cm}

\textbf{{\large Key Generation}}
  
\vspace{0.2 cm}

\begin{itemize}
\item Let there are $n$ signers $ \{S_1,S_2,...,S_n\} $ and $m$ designated verifiers $ \{D_1,D_2,...,D_m\} $
\item Let $M \ \in \ \{0,1\}^* $ be the message, then $H_1(M) \ \in \ G $.
\item Let $a_i$ be the secret key of signer $S_i$ and $p_i \ = \ a_i g $ be the corresponding public key of $S_i \ \ \ \forall \ \ i \in \{1,2,...,n\}  $ known to system $A$.
\item System $A$, then, publishes $u \ = \ \displaystyle\sum_{i=1}^{n} p_i $.
\item Let $b_i$ be the secret key of designated verifier $D_i$ and $q_i \ = \ b_i g$ be the corresponding public key of $D_i \ \ \ \forall \ \ i \in \{1,2,...,m\} $ known to system $B$. 
\item System $B$, then, publishes $v \ = \ \displaystyle\sum_{i=1}^{m} q_i $.
\end{itemize}

\vspace{0.3 cm}

\textbf{{\large Signature Algorithm}}
  
\vspace{0.2 cm}

\begin{itemize}
\item Each signer $S_i$ computes $\sigma_i \ = \ Sign(M,a_i,q_1,q_2,...,q_m) \ = \ e(H_1(M),a_i v) $ as his signature and make it public.
\item After all signatures $(\sigma_1,\sigma_2,...,\sigma_n)$ has been collected by the system $A$, it computes $\sigma \ = \ \displaystyle\prod_{i=1}^{n} \sigma_i $ and send it  to system $B$.
\end{itemize}

\vspace{0.3 cm}

\textbf{{\large Verification Algorithm}}
  
\vspace{0.2 cm}

\begin{itemize}
\item System $B$ has a message $M$ and $\sigma$.
\item Each designated verifier $D_i$ computes $\zeta_i \ = \ Sign(M,b_i,p_1,p_2,...,p_n) \ = \ e(H_1(M),b_i u)$ and make it public to the system $D$.
\item System $B$ computes $\zeta \ = \ \displaystyle\prod_{i=1}^{m} \zeta_i$.
\item Accept the output if $\sigma \ \equiv \ \zeta$ in $G_\tau $, otherwise reject.
\end{itemize}

\vspace{0.3 cm}

\textbf{{\large Transcript-Simulation}}
  
\vspace{0.2 cm}

Since, $\sigma \ \equiv \zeta $ in $G_\tau $, simulated signature is equivalent to the signature produced by system A, thus simulation in this case is obvious .

\vspace{0.3 cm}

\textbf{{\large Example 1}}
  
\vspace{0.2 cm}

Let, $G \ = \ \mathbb{Z}_{11}$ is an additive cyclic group of prime order $11$ and $2$ is a generator of $G$,

$G_\tau \ = \ \{1,2,3,4,6,8,9,12,13,16,18\}  \ \ \ \bigodot_{23} $ is a multiplicative cyclic group of prime order $11$,

$e : \ G \times G \ \rightarrow \ G_\tau$ is symmetric admissible bilinear pairing map such that $e(2,2) \ = \ 2$.

Then, for any $a,b \ \in G , \ \ \exists \ \ 0 \leq x,y \leq 10 $ respectively such that $a \ = \ 2x $ and $b \ = \ 2y $,

 implies $e(a,b) \ = \ e(2,2)^{xy} \ = \ 2^{xy}$
 
Let, $H_1 : \{0,1\}^* \ \rightarrow \ \mathbb{Z}_{11} $ is a cryptographically secure hash function and $M \ \in \ \{0,1\}^* $ be the message, then $H_1(M) \ \in \ \mathbb{Z}_{11}  \ \ \Rightarrow \ \ H_1(M)  \ = \ 2c $ for some $ 1 \leq c \leq 11 $.

Let $A \ = \ \{S_1,S_2,S_3\}$ be three signers and $B \ = \ \{D_1,D_2\}$ be two designated verifiers.

Let $a_1 \ = \ 3$ , $a_2 \ = \ 7$ , $a_3 \ = \ 9$ be the secret keys and $p_1 \ = \ 6$ , $p_2 \ = \ 3$ , $p_3 \ = \ 7$ be public keys (made public to system $A$ only) of $S_1,S_2,S_3$ respectively. Then System $A$ will compute $u \ = \ \displaystyle\sum_{i=1}^{3} p_i \ \equiv \ 5 $ and make it public.

Let $b_1 \ = \ 6$ , $b_2 \ = \ 8$ be the secret keys and $q_1 \ = \ 1$ , $q_2 \ = \ 5$ be public keys (made public to system $B$ only) of $D_1,D_2$ respectively. Then System $B$ will compute $v \ = \ \displaystyle\sum_{i=1}^{2} q_i \ \equiv \ 6 $ and make it public.

Note that $$\sigma_1 \ = \ e(2,2)^{9c} \ = \ 2^{9c}$$
 $$\sigma_2 \ = \ e(2,2)^{10c} \ = \ 2^{10c}$$
 $$\sigma_3 \ = \ e(2,2)^{5c} \ = \ 2^{5c}$$
implies
$$ \sigma \ = \ \displaystyle\prod_{i=1}^{3} \sigma_i  \ = \ 2^{24c} \ \equiv \ 2^{2c} \ \ in \ G_\tau$$

Similarly
$$\zeta_1 \ = \ e(2,2)^{4c} \ = \ 2^{4c} $$
$$\zeta_2 \ = \ e(2,2)^{9c} \ = \ 2^{9c} $$
implies 
$$ \zeta \ = \ \displaystyle\prod_{i=1}^{2} \zeta_i  \ = \ 2^{13c} \ \equiv \ 2^{2c} \ \ in \ G_\tau$$


\vspace{0.3 cm}

\textbf{{\large Example 2}}
  
\vspace{0.2 cm}

Let, $G \ = \ \mathbb{Z}_{53} $ is an additive cyclic group of prime order $53$ and $g \ = \ 5$ is a generator of $G$,

$G_\tau \ = \ < \ 3 \ / \ 3^{53} \ = \ 1>  \ \ \ \bigodot_{107} $ is a multiplicative cyclic group of prime order $53$,

$e : \ G \times G \ \rightarrow \ G_\tau$ is symmetric admissible bilinear pairing map such that $e(5,5) \ = \ 3$.

Then, for any $a,b \ \in G , \ \ \exists \ \ 0 \leq x,y \leq 52 $ respectively such that $a \ = \ 5x $ and $b \ = \ 5y $,

 implies $e(a,b) \ = \ e(5,5)^{xy} \ = \ 3^{xy}$
 
Let, $H_1 : \{0,1\}^* \ \rightarrow \ \mathbb{Z}_{53} $ is a cryptographically secure hash function and $M \ \in \ \{0,1\}^* $ be the message, then $H_1(M) \ \in \ \mathbb{Z}_{53}  \ \ \Rightarrow \ \ H_1(M)  \ = \ 5c $ for some $ 1 \leq c \leq 52 $.

Let $A \ = \ \{S_1,S_2,S_3,S_4,S_5\}$ be the system having five signers and $B \ = \ \{D_1,D_2,D_3,D_4,D_5,D_6,D_7\}$ be the system having seven designated verifiers.

Let $a_1 \ = \ 7$ , $a_2 \ = \ 12$ , $a_3 \ = \ 15$ , $a_4 \ = \ 19$ , $a_5 \ = \ 31$ be the secret keys and $p_1 \ = \ 35$ , $p_2 \ = \ 7$ , $p_3 \ = \ 22$ , $p_4 \ = \ 42$ , $p_5 \ = \ 49$ be public keys (made public to system $A$ only) of $S_1,S_2,S_3,S_4,S_5$ respectively. Then System $A$ will compute $u \ = \ \displaystyle\sum_{i=1}^{5} p_i \ \equiv \ 49 $ and make it public.

Let $b_1 \ = \ 10$ , $b_2 \ = \ 13$ , $b_3 \ = \ 17$ , $b_4 \ = \ 23$ , $b_5 \ = \ 51$ , $b_6 \ = \ 27$ ,  $b_7 \ = \ 36$  be the secret keys and $q_1 \ = \ 50$ , $q_2 \ = \ 12$ , $q_3 \ = \ 32$ , $q_4 \ = \ 9$ , $q_5 \ = \ 43$ , $q_6 \ = \ 29$ , $q_2 \ = \ 21$ be public keys (made public to system $B$ only) of $D_1,D_2$ respectively. Then System $B$ will compute $v \ = \ \displaystyle\sum_{i=1}^{6} q_i \ \equiv \ 37 $ and make it public.

Note that $$\sigma_1 \ = \ 3^{20c}$$
 $$\sigma_2 \ = \ 3^{4c}$$
 $$\sigma_3 \ = \ 3^{5c}$$
 $$\sigma_4 \ = \ 3^{24c}$$
 $$\sigma_5 \ = \ 3^{28c}$$
implies
$$ \sigma \ = \ \displaystyle\prod_{i=1}^{5} \sigma_i  \ \equiv \ 3^{28c} \ \ in \ G_\tau$$

Similarly
$$\zeta_1 \ = \ 3^{45c} $$
$$\zeta_2 \ = \ 3^{32c} $$
$$\zeta_3 \ = \ 3^{50c} $$
$$\zeta_4 \ = \ 3^{24c} $$
$$\zeta_5 \ = \ 3^{44c} $$
$$\zeta_6 \ = \ 3^{42c} $$
$$\zeta_7 \ = \ 3^{3c} $$
implies 
$$ \zeta \ = \ \displaystyle\prod_{i=1}^{7} \zeta_i  \ \equiv \ 3^{28c} \ \ in \ G_\tau$$


\subsection{Algorithm 2}

\vspace{0.3 cm}

\hspace{0.4 cm} \textbf{{\large Set up}}
  
\vspace{0.2 cm}

Let $K$ be the security parameter and $sp \ = \ (\widetilde{G},g_A,g_B,p,n_A,n_B,H_2) \ \leftarrow \ Setup(K)$, where
\begin{itemize}

\item $p$ be a large prime such that $ n_A $ and $ n_B $ are factors of $p-1$,

\item $\widetilde{G} \ = \ Z_p ^*$ and $g_A \ , \ g_B \ \ \in \ \widetilde{G}$ such that $|g_A| \ = \ n_A $, $|g_B| \ = \ n_B $,

\item $H_2$ be a cryptographically secured hash function with arbitrary bit length.

\end{itemize} 


\vspace{0.3 cm}

\textbf{{\large Key Generation}}
  
\vspace{0.2 cm}

\emph{\textbf{Key generation algorithm for system $A$:}}

\begin{enumerate}
\item System $A$ chooses prime $p_A$ and $q_A$, then computes $n_A \ = \ p_A \ * \ q_A$.
\item Each $S_i$ follow the following steps :
\begin{itemize}
\item randomly choose $e_{A_i}$ such that g.c.d.$(e_{A_i}, \phi(n_A)) \ = \ 1$.
\item compute $d_{A_i}$ such that $e_{A_i} \ d_{A_i} \ = \ 1 \ (mod \ \phi(n_A))$.
\item choose an integer $x_{A_i} \ \in \ Z_p ^* $.
\item calculate $y_{A_i} \ = \ g_A ^{x_{A_i}} \ (mod \ p)$.
\item public key of $S_i$ is $(e_{A_i}, y_{A_i})$.
\item private key of $S_i$ is $(d_{A_i}, x_{A_i})$.
\end{itemize}
\item Note : Since, each member in system $A$ knows $\phi(n_A)$ and $e_{A_i} \ \ \forall i $, they can compute $d_{A_i}$. Thus, $d_{A_i}$ is private for members of system $B$ but not for members of system $A$.
\end{enumerate}


\emph{\textbf{Key generation algorithm for system $B$:}}
\begin{enumerate}
\item System $B$ chooses prime $p_B$ and $q_B$, then computes $ n_B \ = \ p_B \ * \ q_B $.
\item Each $D_i$ follow the following steps :
\begin{itemize}
\item randomly choose $e_{B_i}$ such that g.c.d.$(e_{B_i}, \phi(n_B)) \ = \ 1$.
\item compute $d_{B_i}$ such that $e_{B_i} \ d_{B_i} \ = \ 1 \ (mod \ \phi(n_B))$.
\item choose an integer $x_{B_i} \ \in \ Z_p ^* $.
\item calculate $y_{B_i} \ = \ g_B ^{x_{B_i}} \ (mod \ p)$.
\item public key of $D_i$ is $(e_{B_i}, y_{B_i})$.
\item private key of $D_i$ is $(d_{B_i}, x_{B_i})$.
\end{itemize}
\item Note : Since, each member in system $B$ knows $\phi(n_B)$ and $e_{B_i} \ \ \forall i $, they can compute $d_{B_i}$. Thus, $d_{B_i}$ is private for members of system $A$ but not for members of system $B$.
\end{enumerate}


\vspace{0.3 cm}

\textbf{{\large Signature Algorithm}}
  
\vspace{0.2 cm}

System $A$ generates a signature $(r,s,t,\bar{u})$ for a message $M$ as follows:
\begin{enumerate}
\item Each $S_i$ chooses a random integer $k_i$ and keep it secret.
\item Each $S_i$ generates $r_i \ = \ g_A ^{k_i} (y_{B_1} \ y_{B_2} \ . \ . \ . \ y_{B_m} )^{- k_i} \ \ (mod \ p) $ ,  $ s_i \ = \ g_B ^{k_i} \ \ (mod  \ p) $ and $ w_i \ = \ g_A ^{k_i} \ \ (mod  \ p) $ and make it public to the system $A$.
\item System computes :
\begin{itemize}
\item $r \ = \ \displaystyle \prod_{i=1}^{n} r_i \ \ (mod  \ p) $
\item $s \ = \ (\displaystyle \prod_{i=1}^{n} s_i)^r \ \ (mod  \ p) $
\item $w \ = \ (\displaystyle \prod_{i=1}^{n} w_i)^r \ \ (mod  \ p) $
\item $ z \ = \ H_2(M,w) $
\item $ t \ = \ z^{\displaystyle \prod_{i=1}^{m} e_{B_i}}  \ \ (mod \ n_B)$ 
\end{itemize} 
\item After computing $(r,s,w,z,t)$, they are made public to all $S_i$ and then each $S_i$ computes $v_i \ = \ z x_{A_i} \ + \ k_i r$ and make it public to the system $A$.
\item System computes
\begin{itemize}
\item $\bar{v} \ = \ \displaystyle \sum_{i=1}^{n} v_i \ \ (mod \ n_A) $
\item $\bar{u} \ = \ (\bar{v})^{\displaystyle \prod_{i=1}^{n} d_{A_i}} \ \ (mod \ n_A) $
\end{itemize}
System $A$ sends the message $M$ and the signature $(r,s,t,\bar{u})$ to the system $B$ (and all the designated verifiers $D_i$).
\end{enumerate}


\vspace{0.3 cm}

\textbf{{\large Verification Algorithm}}
  
\vspace{0.2 cm}

Upon receiving the message $M$ and signature $(r,s,t,\bar{u})$, each designated verifier $D_i$ computes $z_i \ = \ s^{x_{B_i}}$ and make it public. Then, each $D_i$ can separately or together verify and validate the signature by checking the following equations:
\begin{enumerate}
\item Computes $a \ = \ (\bar{u})^{\displaystyle \prod_{i=1}^{n} e_{A_i}} \ \ (mod \ n_A) $ 
\item Computes $b \ = \ t^{\displaystyle \prod_{i=1}^{m} d_{B_i}} \ \ (mod \ n_B) $
\item Check $ c \ = \ g_A ^a  \ (y_{A_1} y_{A_2} . . . y_{A_n})^{-b} \ \ (mod \ p) \ \ = \ \ r^r \displaystyle \prod_{i=1}^{m} z_i \ \ (mod \ p) $
\item Lastly, signatures are accepted if $b \ = \ H_2(M,c)$.
\item If any of the condition from above two conditions fail, verifier can deny to accept the signature. In case of more than $2$ verifiers, if $50 \% $ or more verifiers deny the signature, system $B$ will return the signature to system $A$.
\end{enumerate}


\vspace{0.3 cm}

\textbf{{\large Transcript-Simulation}}
  
\vspace{0.2 cm}

System $B$ (or designated verifier $D$) can simulate the correct transcripts as follows:
\begin{enumerate}
\item System $B$ chooses a random integer $k^{'}$ and keep it secret.
\item Then $D$ generates :

\begin{itemize}
\item $r^{'} \ = \ g_B ^{k^{'}} (y_{A_1} \ y_{A_2} \ . \ . \ . \ y_{A_n} )^{- k^{'}} \ \ (mod \ p) $
\item $s^{'} \ = \ g_A ^{k^{'} r^{'}} \ \ (mod  \ p) $
\item $w^{'} \ = \ g_B ^{k^{'} r^{'}} \ \ (mod  \ p) $
\item $ z^{'} \ = \ H_2(M,w^{'}) $
\item $ t^{'} \ = \ {z^{'}}^{\displaystyle \prod_{i=1}^{n} e_{A_i}}  \ \ (mod \ n_A)$ 
\item $v^{'} \ = \ \displaystyle \sum_{i=1}^{m} {v_i}^{'} \ \ (mod \ n_B) $, where $ {v_i}^{'} \ = \ z^{'} x_{B_i} \ + \ k^{'} r^{'} $ is made public by $D_i$
\item $u^{'} \ = \ ({v^{'}})^{d_{B}} \ \ (mod \ n_B) $, where $d_B \ = \ \displaystyle \prod_{i=1}^{m} d_{B_i}$
\end{itemize} 
System $B$ simulated signature is $(r^{'},s^{'},t^{'},u^{'})$.
\end{enumerate}



\vspace{0.3 cm}

\textbf{{\large Example 1}}
  
\vspace{0.2 cm}

Let $A = \{ S_1,S_2,S_3 \} $ be the group of three signers (or provers) and $B = \{ D_1, D_2 \} $ be the group of two designated verifiers.

Let $p \ = \ 211 $ be a prime such that $ n_A \ = \ 15 $ and $ n_B \ =  \ 14 $ are factors of $p-1 \ =  \ 210 $.

Let $\widetilde{G} \ = \ Z_{211} ^*$ and $g_A \ = \ 137 \ , \ g_B \ = \ 63 $  in  $ \widetilde{G} $  such  that $|137| \ = \ 15 $, $|63| \ = \ 14 $.

Let $H_2$ be a cryptographically secured hash function with arbitrary bit length.


\textbf{{ Key generation algorithm for system $A$:}}
\begin{enumerate}
\item System $A$ chooses prime $p_A \ = \ 3 $ and $q_A \ = \ 5 $, then computes $n_A \ = \ 3 \ * \ 5 $.
\item Each $S_i$ follow the following steps :

\begin{center}
\begin{tabular}{ |c|c|c|c|c|c| } 
 \hline
 $S_i$ & $e_{A_i}$ & $d_{A_i}$ & $x_{A_i}$ & $y_{A_i}$ \\
 \hline \hline 
 $S_1$ & $13$ & $5$ & $7$ & $150$ \\
 $S_2$ & $7$ & $7$ & $16$ & $137$ \\
 $S_3$ & $11$ & $3$ & $21$ & $55$ \\
 \hline
\end{tabular}
\end{center}

\item public key of $S_i$ is $(e_{A_i}, y_{A_i})$.
\item private key of $S_i$ is $(d_{A_i}, x_{A_i})$.

\item Note : Since, each member in system $A$ knows $\phi(n_A)$ and $e_{A_i} \ \ \forall i $, they can compute $d_{A_i}$. Thus, $d_{A_i}$ is private for members of system $B$ but not for members of system $A$.
\end{enumerate}


\textbf{{ Key generation algorithm for system $B$:}}
\begin{enumerate}
\item System $B$ chooses prime $p_B \ = \ 2 $ and $q_B \ = \ 7 $, then computes $ n_B \ = \ 2 \ * \ 7 $.
\item Each $D_i$ follow the following steps :

\begin{center}
\begin{tabular}{ |c|c|c|c|c|c| } 
 \hline
 $D_i$ & $e_{B_i}$ & $d_{B_i}$ & $x_{B_i}$ & $y_{B_i}$ \\
 \hline \hline 
 $D_1$ & $5$ & $5$ & $19$ & $153$ \\
 $D_2$ & $11$ & $5$ & $17$ & $12$ \\
 \hline
\end{tabular}
\end{center}

\item public key of $D_i$ is $(e_{B_i}, y_{B_i})$.
\item private key of $D_i$ is $(d_{B_i}, x_{B_i})$.

\item Note : Since, each member in system $B$ knows $\phi(n_B)$ and $e_{B_i} \ \ \forall i $, they can compute $d_{B_i}$. Thus, $d_{B_i}$ is private for members of system $A$ but not for members of system $B$.
\end{enumerate}


\vspace{0.3 cm}

\textbf{{\large Signature Algorithm}}
  
\vspace{0.2 cm}

\begin{enumerate}
\item Each $S_i$ chooses a random integer $k_i$ and keep it secret and then, generates $$ r_i \ = \ g_A ^{k_i} (y_{B_1} \ y_{B_2} \ . \ . \ . \ y_{B_e} )^{- k_i} \ \ (mod \ p) \ , \    s_i \ = \ g_B ^{k_i} \ \ (mod  \ p) \ and \  w_i \ = \ g_A ^{k_i} \ \ (mod  \ p) $$ and make it public to the system $A$.

\begin{center}
\begin{tabular}{ |c|c|c|c|c|c| } 
 \hline
 $S_i$ & $k_i$ & $r_i$ & $s_i$ & $w_i$ \\
 \hline \hline 
 $S_1$ & $8$ & $136$ & $148$ & $83$ \\
 $S_2$ & $12$ & $114$ & $171$ & $71$ \\
 $S_3$ & $14$ & $134$ & $210$ & $134$ \\
 \hline
\end{tabular}
\end{center}

\item System computes :
\begin{itemize}
\item $r \ = \ \displaystyle \prod_{i=1}^{3} r_i  \ \equiv \ 30 \ \ (mod  \ p) $
\item $s \ = \ (\displaystyle \prod_{i=1}^{3} s_i)^r \ \equiv \ 144 \ \ (mod  \ p) $
\item $w \ = \ (\displaystyle \prod_{i=1}^{3} w_i)^r \ \equiv \ 1 \ \ (mod  \ p) $
\item $ z \ = \ H_2(M,w) $
\item $ t \ = \ z^{\displaystyle \prod_{i=1}^{2} e_{B_i}} \ \equiv \ z^{55} \ \equiv z  \ \ (mod \ 14)$ 
\end{itemize} 
\item After computing $(r,s,w,z,t)$, they are made public to all $P_i$ and then each $P_i$ computes $v_i \ = \ z x_{A_i} \ + \ k_i r$ and make it public to the system $A$.
$$ v_1 \ = \ 7z + 240 \ , \ v_2 \ = \ 16z + 360 \ , \ v_3 \ = \ 21z + 420 $$
\item System computes
\begin{itemize}
\item $\bar{v} \ = \ \displaystyle \sum_{i=1}^{3} v_i \ \equiv 44z \ \ (mod \ 15) $
\item $\bar{u} \ = \ (\bar{v})^{\displaystyle \prod_{i=1}^{3} d_{A_i}} \ \equiv \ (\bar{v})^{5*7*3} \ \equiv \ \bar{v} \ \ (mod \ 15) $
\end{itemize}
System $A$ sends the message $M$ and the signature $(r,s,t,\bar{u})$ to the system $B$ (and all the verifiers $D_i$).
\end{enumerate}


\vspace{0.3 cm}

\textbf{{\large Verification Algorithm}}
  
\vspace{0.2 cm}

Upon receiving the message $M$ and the signature $(r,s,t,\bar{u})$, each designated verifier $D_i$ computes $z_i \ = \ s^{x_{B_i}}$ and make it public. Then, each $D_i$ can separately or together verify and validate the signature by checking the following equations:
\begin{enumerate}
\item Computes $a \ = \ (\bar{u})^{\displaystyle \prod_{i=1}^{3} e_{A_i}} \ \ (mod \ n_A) $ 
\item Computes $b \ = \ t^{\displaystyle \prod_{i=1}^{2} d_{B_i}} \ \ (mod \ n_B) $
\item Check $ c \ = \ g_A ^a  \ (y_{A_1} y_{A_2} y_{A_3})^{-b} \ \ (mod \ p) \ \ = \ \ r^r \displaystyle \prod_{i=1}^{2} z_i \ \ (mod \ p) $
\item Lastly, signatures are accepted if $b \ = \ H_2(M,c)$.
\item If any of the condition from above two conditions fail, verifier can deny to accept the signature. In case of more than $2$ verifiers, if $50 \% $ or more verifiers deny the signature, system $B$ will return the signature to system $A$.
\end{enumerate}


\vspace{0.3 cm}

\textbf{{\large Example 2}}
  
\vspace{0.2 cm}

Let $A = \{ S_1,S_2,S_3,S_4,S_5 \} $ be the group of five signers (or provers) and $B = \{ D_1, D_2, D_3, D_4, D_5, D_6, D_7 \} $ be the group of seven designated verifiers.

Let $p \ = \ 102103 $ be a prime such that $ n_A \ = \ 91 $ and $ n_B \ =  \ 187 $ are factors of $p-1 \ =  \ 102102 $.

Let $\widetilde{G} \ = \ Z_{102103} ^*$ and $g_A \ = \ 44494 \ , \ g_B \ = \ 12733 $  in  $ \widetilde{G} $  such  that $|44494| \ = \ 91 $, $|12733| \ = \ 187 $.

Let $H_2$ be a cryptographically secured hash function with arbitrary bit length.


\textbf{{ Key generation algorithm for system $A$:}}
\begin{enumerate}
\item System $A$ chooses prime $p_A \ = \ 7 $ and $q_A \ = \ 13 $, then computes $n_A \ = \ 7 \ * \ 13 $.
\item Each $S_i$ follow the following steps :

\begin{center}
\begin{tabular}{ |c|c|c|c|c|c| } 
 \hline
 $S_i$ & $e_{A_i}$ & $d_{A_i}$ & $x_{A_i}$ & $y_{A_i}$ \\
 \hline \hline 
 $S_1$ & $5$ & $29$ & $15$ & $58327$ \\
 $S_2$ & $11$ & $59$ & $19$ & $69494$ \\
 $S_3$ & $7$ & $31$ & $24$ & $69058$ \\
 $S_4$ & $23$ & $47$ & $18$ & $88515$ \\
 $S_5$ & $35$ & $35$ & $32$ & $26123$ \\
 \hline
\end{tabular}
\end{center}

\item public key of $S_i$ is $(e_{A_i}, y_{A_i})$.
\item private key of $S_i$ is $(d_{A_i}, x_{A_i})$.

\item Note : Since, each member in system $A$ knows $\phi(n_A)$ and $e_{A_i} \ \ \forall i $, they can compute $d_{A_i}$. Thus, $d_{A_i}$ is private for members of system $B$ but not for members of system $A$.
\end{enumerate}


\textbf{{ Key generation algorithm for system $A$:}}
\begin{enumerate}
\item System $B$ chooses prime $p_B \ = \ 11 $ and $q_B \ = \ 17 $, then computes $ n_B \ = \ 11 \ * \ 17 $.
\item Each $D_i$ follow the following steps :

\begin{center}
\begin{tabular}{ |c|c|c|c|c|c| } 
 \hline
 $D_i$ & $e_{B_i}$ & $d_{B_i}$ & $x_{B_i}$ & $y_{B_i}$ \\
 \hline \hline 
 $D_1$ & $3$ & $107$ & $32$ & $37552$ \\
 $D_2$ & $7$ & $23$ & $17$ & $64089$ \\
 $D_3$ & $11$ & $131$ & $22$ & $63449$ \\
 $D_4$ & $19$ & $59$ & $27$ & $93579$ \\
 $D_5$ & $51$ & $91$ & $18$ & $38061$ \\
 $D_6$ & $27$ & $83$ & $21$ & $91435$ \\
 $D_7$ & $91$ & $51$ & $51$ & $30671$ \\
 \hline
\end{tabular}
\end{center}

\item public key of $D_i$ is $(e_{B_i}, y_{B_i})$.
\item private key of $D_i$ is $(d_{B_i}, x_{B_i})$.

\item Note : Since, each member in system $B$ knows $\phi(n_B)$ and $e_{B_i} \ \ \forall i $, they can compute $d_{B_i}$. Thus, $d_{B_i}$ is private for members of system $A$ but not for members of system $B$.
\end{enumerate}


\vspace{0.3 cm}

\textbf{{\large Signature Algorithm}}
  
\vspace{0.2 cm}

\begin{enumerate}
\item Each $S_i$ chooses a random integer $k_i$ and keep it secret and then, generates $$ r_i \ = \ g_A ^{k_i} (y_{B_1} \ y_{B_2} \ . \ . \ . \ y_{B_e} )^{- k_i} \ \ (mod \ p) \ , \    s_i \ = \ g_B ^{k_i} \ \ (mod  \ p) \ and \  w_i \ = \ g_A ^{k_i} \ \ (mod  \ p) $$ and make it public to the system $A$.

\begin{center}
\begin{tabular}{ |c|c|c|c|c|c| } 
 \hline
 $S_i$ & $k_i$ & $r_i$ & $s_i$ & $w_i$ \\
 \hline \hline 
 $S_1$ & $42980$ & $22513$ & $68227$ & $59022$ \\
 $S_2$ & $68841$ & $77234$ & $67990$ & $84473$ \\
 $S_3$ & $82718$ & $60319$ & $8171$ & $15368$ \\
 $S_4$ & $90739$ & $49375$ & $58517$ & $68284$ \\
 $S_5$ & $19344$ & $40471$ & $35552$ & $91619$ \\
 \hline
\end{tabular}
\end{center}

\item System computes :
\begin{itemize}
\item $r \ = \ \displaystyle \prod_{i=1}^{5} r_i  \ \equiv \ 41707 \ \ (mod  \ p) $
\item $s \ = \ (\displaystyle \prod_{i=1}^{5} s_i)^r \ \equiv \ 90653 \ \ (mod  \ p) $
\item $w \ = \ (\displaystyle \prod_{i=1}^{5} w_i)^r \ \equiv \ 91371 \ \ (mod  \ p) $
\item $ z \ = \ H_2(M,w) $
\item $ t \ = \ z^{\displaystyle \prod_{i=1}^{7} e_{B_i}} \ \equiv \ z^{549972423} \ \equiv z^{103}  \ \ (mod \ 187)$, provided $g.c.d.(z,187) \ = \ 1 $
\end{itemize} 
\item After computing $(r,s,w,z,t)$, they are made public to all $S_i$ and then each $S_i$ computes $v_i \ = \ z x_{A_i} \ + \ k_i r$ and make it public to the system $A$.
$$ v_1 \ = \ 15z + 42980r \ , \ v_2 \ = \ 19z + 68841r \ , \ v_3 \ = \ 24z + 82718r \ , \ v_4\ = \ 18z + 90739r  \ , \ v_5 \ = \ 32z + 19344r  $$
\item System computes
\begin{itemize}
\item $\bar{v} \ = \ \displaystyle \sum_{i=1}^{5} v_i \ \equiv 17z + 31 \ \ (mod \ 91) $
\item $\bar{u} \ = \ (\bar{v})^{\displaystyle \prod_{i=1}^{5} d_{A_i}} \ \equiv \ (\bar{v})^{87252445} \ \equiv \ (\bar{v})^{37} \ \ (mod \ 91) $
\end{itemize}
System $A$ sends the message $M$ and the signature $(r,s,t,\bar{u})$ to the system $B$ (and all the verifiers $D_i$).
\end{enumerate}


\vspace{0.3 cm}

\textbf{{\large Verification Algorithm}}
  
\vspace{0.2 cm}

Upon receiving the message $M$ and the signature $(r,s,t,\bar{u})$, each verifier $D_i$ computes $z_i \ = \ s^{x_{B_i}}$ and make it public. Then, each $D_i$ can separately or together verify and validate the signature by checking the following equations:
\begin{enumerate}
\item Computes $a \ = \ (\bar{u})^{\displaystyle \prod_{i=1}^{5} e_{A_i}} \ \ (mod \ n_A) $ 
\item Computes $b \ = \ t^{\displaystyle \prod_{i=1}^{7} d_{B_i}} \ \ (mod \ n_B) $
\item Check $ c \ = \ g_A ^a  \ (y_{A_1} y_{A_2} y_{A_3} y_{A_4} y_{A_5})^{-b} \ \ (mod \ p) \ \ = \ \ r^r \displaystyle \prod_{i=1}^{7} z_i \ \ (mod \ p) $
\item Lastly, signatures are accepted if $b \ = \ H_2(M,c)$.
\item If any of the condition from above two conditions fail, verifier can deny to accept the signature. In case of more than $2$ verifiers, if $50 \% $ or more verifiers deny the signature, system $B$ will return the signature to system $A$.
\end{enumerate}

\subsection{Algorithm 3}


\vspace{0.3 cm}

\textbf{{\large Set up}}
  
\vspace{0.2 cm}

Let $K$ be the security parameter and $sp \ = \ (E,P,Q,p,n_A,n_B,H_3) \ \leftarrow \ Setup(K)$, where

\begin{itemize}
\item $p$ be a large prime,
\item $E$ be an elliptic curve over $Z_p$ such that $ n_A $ and $ n_B $ are factors of $|E|$, where $n_A$ and $n_B$ are product of two large primes,
\item $P \ , \ Q \ \ \in \ E $ such that $|P| \ = \ n_A $, $|Q| \ = \ n_B $,
\item $H_3$ be a cryptographically secured hash function with arbitrary bit length.

\end{itemize}


\vspace{0.3 cm}

\textbf{{\large Key Generation}}
  
\vspace{0.2 cm}

\emph{\textbf{{ Key generation algorithm for system $A$:}}}
\begin{enumerate}
\item System $A$ chooses prime $p_A$ and $q_A$, then computes $n_A \ = \ p_A \ * \ q_A$.
\item Each $S_i$ follow the following steps :
\begin{itemize}
\item randomly choose $e_{A_i}$ such that g.c.d.$(e_{A_i}, \phi(n_A)) \ = \ 1$.
\item compute $d_{A_i}$ such that $e_{A_i} \ d_{A_i} \ = \ 1 \ (mod \ \phi(n_A))$.
\item choose an integer $x_{A_i} \ \in \ Z_p ^* $.
\item calculate $y_{A_i} \ = \ [x_{A_i}]P \ (mod \ p)$.
\item public key of $S_i$ is $(e_{A_i}, y_{A_i})$.
\item private key of $S_i$ is $(d_{A_i}, x_{A_i})$.
\end{itemize}
\item Note : Since, each member in system $A$ knows $\phi(n_A)$ and $e_{A_i} \ \ \forall i $, they can compute $d_{A_i}$. Thus, $d_{A_i}$ is private for members of system $B$ but not for members of system $A$.
\end{enumerate}


\emph{\textbf{{ Key generation algorithm for system $B$:}}}
\begin{enumerate}
\item System $B$ chooses prime $p_B$ and $q_B$, then computes $ n_B \ = \ p_B \ * \ q_B $.
\item Each $D_i$ follow the following steps :
\begin{itemize}
\item randomly choose $e_{B_i}$ such that g.c.d.$(e_{B_i}, \phi(n_B)) \ = \ 1$.
\item compute $d_{B_i}$ such that $e_{B_i} \ d_{B_i} \ = \ 1 \ (mod \ \phi(n_B))$.
\item choose an integer $x_{B_i} \ \in \ Z_p ^* $.
\item calculate $y_{B_i} \ = \ [x_{B_i}]Q \ (mod \ p)$.
\item public key of $D_i$ is $(e_{B_i}, y_{B_i})$.
\item private key of $D_i$ is $(d_{B_i}, x_{B_i})$.
\end{itemize}
\item Note : Since, each member in system $B$ knows $\phi(n_B)$ and $e_{B_i} \ \ \forall i $, they can compute $d_{B_i}$. Thus, $d_{B_i}$ is private for members of system $A$ but not for members of system $B$.
\end{enumerate}


\vspace{0.3 cm}

\textbf{{\large Signature Algorithm}}
  
\vspace{0.2 cm}

System $A$ generates a signature $(r,s,t,\bar{u})$ for a message $M$ as follows:
\begin{enumerate}
\item Each $S_i$ chooses a random integer $k_i$ and keep it secret.
\item Each $S_i$ generates $r_i \ = \ [{k_i}]P -[{k_i}](y_{B_1} \ + \ y_{B_2} \ + \ . \ . \ . \ + \ y_{B_m} ) \ \ (mod \ p) $ ,  $ s_i \ = \ [{k_i}]Q \ \ (mod  \ p) $ and $ w_i \ = \ [{k_i}]P \ \ (mod  \ p) $ and make it public to the system $A$.
\item System computes :
\begin{itemize}
\item $r \ = \ \displaystyle \sum_{i=1}^{n} r_i \ \ (mod  \ p) $
\item $s \ = \ \displaystyle \sum_{i=1}^{n} s_i \ \ (mod  \ p) $
\item $w \ = \ \displaystyle \sum_{i=1}^{n} w_i \ \ (mod  \ p) $
\item $ z \ = \ H_3(M,w) $
\item $ t \ = \ z^{\displaystyle \prod_{i=1}^{m} e_{B_i}}  \ \ (mod \ n_B)$ 
\end{itemize} 
\item After computing $(r,s,w,z,t)$, they are made public to all $S_i$ and then each $S_i$ computes $v_i \ = \ z . x_{A_i} \ + \ k_i$ and make it public to the system $A$.
\item System computes
\begin{itemize}
\item $\bar{v} \ = \ \displaystyle \sum_{i=1}^{n} v_i \ \ (mod \ n_A) $
\item $\bar{u} \ = \ (\bar{v})^{\displaystyle \prod_{i=1}^{n} d_{A_i}} \ \ (mod \ n_A) $
\end{itemize}
System $A$ sends the message $M$ and the signature $(r,s,t,\bar{u})$ to the system $B$ (and all the verifiers $D_i$).
\end{enumerate}



\vspace{0.3 cm}

\textbf{{\large Verification Algorithm}}
  
\vspace{0.2 cm}

Upon receiving the message $M$ and the signature $(r,s,t,\bar{u})$, each verifier $D_i$ computes $z_i \ = \ [x_{B_i}]s$ and make it public. Then, each $D_i$ can separately or together verify and validate the signature by checking the following equations:
\begin{enumerate}
\item Computes $a \ = \ (\bar{u})^{\displaystyle \prod_{i=1}^{n} e_{A_i}} \ \ (mod \ n_A) $ 
\item Computes $b \ = \ t^{\displaystyle \prod_{i=1}^{m} d_{B_i}} \ \ (mod \ n_B) $
\item Check $ c \ = \ [a]P \ - \ b(y_{A_1} \ + \ y_{A_2} \ + \  . . . \ + \ y_{A_n}) \ \ (mod \ p) \ \ = \ \ r \ + \ \displaystyle \sum_{i=1}^{m} z_i \ \ (mod \ p) $
\item Lastly, signatures are accepted if $b \ = \ H_3(M,c)$.
\item If any of the condition from above two conditions fail, verifier can deny to accept the signature. In case of more than $2$ verifiers, if $50 \% $ or more verifiers deny the signature, system $B$ will return the signature to system $A$.
\end{enumerate}


\vspace{0.3 cm}

\textbf{{\large Transcript Simulation}}
  
\vspace{0.2 cm}

System $B$ (or designated verifier $D$) can simulate the correct transcripts as follows:
\begin{enumerate}
\item System $B$ chooses a random integer $k^{'}$ and keep it secret.
\item Then $B$ generates :
\begin{itemize}
\item $r^{'} \ = \ [{k^{'}}]Q - [{k^{'}}] (y_{A_1} \ + \ y_{A_2} \ + \ . \ . \ . \ + \ y_{A_n} ) \ \ (mod  \ p) $
\item $s^{'} \ = \ [{k^{'}}]P \ \ (mod  \ p) $
\item $w^{'} \ = \ [{k^{'}}]Q \ \ (mod  \ p) $
\item $ z^{'} \ = \ H(m,w^{'}) $
\item $ t^{'} \ = \ {z^{'}}^{(\displaystyle \prod_{i=1}^{n} e_{A_i})}  \ \ (mod \ n_A)$ 
\item $v^{'} \ = \ \displaystyle \sum_{i=1}^{m} {v_i}^{'} \ \ (mod \ n_B) $, where $ {v_i}^{'} \ = \ z^{'} x_{B_i} \ + \ k^{'} $ is made public by $D_i$

\item $u^{'} \ = \ ({v^{'}})^{d_{B}} \ \ (mod \ n_B) $, where $d_B \ = \ \displaystyle \prod_{i=1}^{m} d_{B_i}$
\end{itemize} 
System $B$ simulated signature is $(r^{'},s^{'},t^{'},u^{'})$.
\end{enumerate}


\vspace{0.3 cm}

\textbf{{\large Example 1}}
  
\vspace{0.2 cm}

Let $A = \{ S_1,S_2,S_3 \} $ be the group of three signers (or provers) and $B = \{ D_1, D_2 \} $ be the group of two designated verifiers.

Let $p \ = \ 419 $ be a prime and $E$ be an elliptic curve over $\mathbb{Z}_p$ such that 
$$E \ = \ E_{\mathbb{Z}_p} \ = \ \{ (x,y) \in \mathbb{Z}_p \times \mathbb{Z}_p \ / \ y^2 = x^3 + 2 \}$$

Note that $|E| \ = \ 420$ and $E$ is a cyclic group . Also $ n_A \ = \ 15 $ and $ n_B \ =  \ 14 $ are factors of $|E| $.

Let $P \ = \ (22,151)  \ , \ Q \ = \ (55,156) $  in  $ E $  such  that $|P| \ = \ 15 $, $|Q| \ = \ 14 $.

Let $H_3$ be a cryptographically secured hash function with arbitrary bit length.


\textbf{{ Key generation algorithm for system $A$:}}
\begin{enumerate}
\item System $A$ chooses prime $p_A \ = \ 3 $ and $q_A \ = \ 5 $, then computes $n_A \ = \ 3 \ * \ 5 $.
\item Each $S_i$ follow the following steps :

\begin{center}
\begin{tabular}{ |c|c|c|c|c|c| } 
 \hline
 $S_i$ & $e_{A_i}$ & $d_{A_i}$ & $x_{A_i}$ & $y_{A_i}$ \\
 \hline \hline 
 $S_1$ & $13$ & $5$ & $7$ & $7P$ \\
 $S_2$ & $7$ & $7$ & $16$ & $P$ \\
 $S_3$ & $11$ & $3$ & $21$ & $6P$ \\
 \hline
\end{tabular}
\end{center}

\item public key of $S_i$ is $(e_{A_i}, y_{A_i})$.
\item private key of $S_i$ is $(d_{A_i},x_{A_i})$.

\item Note : Since, each member in system $A$ knows $\phi(n_A)$ and $e_{A_i} \ \ \forall i $, they can compute $d_{A_i}$. Thus, $d_{A_i}$ is private for members of system $B$ but not for members of system $A$.
\end{enumerate}


\textbf{{ Key generation algorithm for system $B$:}}
\begin{enumerate}
\item System $B$ chooses prime $p_B \ = \ 2 $ and $q_B \ = \ 7 $, then computes $ n_B \ = \ 2 \ * \ 7 $.
\item Each $D_i$ follow the following steps :

\begin{center}
\begin{tabular}{ |c|c|c|c|c|c| } 
 \hline
 $D_i$ & $e_{B_i}$ & $d_{B_i}$ & $x_{B_i}$ & $y_{B_i}$ \\
 \hline \hline 
 $D_1$ & $5$ & $5$ & $19$ & $5Q$ \\
 $D_2$ & $11$ & $5$ & $17$ & $3Q$ \\
 \hline
\end{tabular}
\end{center}

\item public key of $D_i$ is $(e_{B_i}, y_{B_i})$.
\item private key of $D_i$ is $(d_{B_i},x_{B_i})$.

\item Note : Since, each member in system $B$ knows $\phi(n_B)$ and $e_{B_i} \ \ \forall i $, they can compute $d_{B_i}$. Thus, $d_{B_i}$ is private for members of system $A$ but not for members of system $B$.
\end{enumerate}


\vspace{0.3 cm}

\textbf{{\large Signature Algorithm}}
  
\vspace{0.2 cm}

\begin{enumerate}
\item Each $S_i$ chooses a random integer $k_i$ and keep it secret and then, generates $$r_i \ = [{k_i}]P -[{k_i}](y_{B_1} \ + \ y_{B_2} \ + \ . \ . \ . \ + \ y_{B_e} ) \ \ (mod \ p) \ \ , \ \  s_i \ = \ [{k_i}]Q \ \ (mod  \ p)  \ \ and \  \  w_i \ = \ [{k_i}]P \ \ (mod  \ p) $$ and make it public to the system $A$ :

\begin{center}
\begin{tabular}{ |c|c|c|c|c|c| } 
 \hline
 $S_i$ & $k_i$ & $r_i$ & $s_i$ & $w_i$ \\
 \hline \hline 
 $S_1$ & $8$ & $8P-8Q$ & $8Q$ & $8P$ \\
 $S_2$ & $12$ & $12P-12Q$ & $12Q$ & $12P$ \\
 $S_3$ & $14$ & $14P$ & $14Q$ & $14P$ \\
 \hline
\end{tabular}
\end{center}

\item System computes :
\begin{itemize}
\item $r \ = \ \displaystyle \sum_{i=1}^{3} r_i  \ \equiv \ 4P-6Q \ \ (mod  \ p) $
\item $s \ = \ \displaystyle \sum_{i=1}^{3} s_i \ \equiv \ 6Q \ \ (mod  \ p) $
\item $w \ = \ \displaystyle \prod_{i=1}^{3} w_i \ \equiv \ 4P \ \ (mod  \ p) $
\item $ z \ = \ H_3(M,w) $
\item $ t \ = \ z^{(\displaystyle \prod_{i=1}^{2} e_{B_i})} \ \equiv \ z^{55} \ \equiv z  \ \ (mod \ 14)$ 
\end{itemize} 
\item After computing $(r,s,w,z,t)$, they are made public to all $S_i$ and then each $S_i$ computes $v_i \ = \ z x_{A_i} \ + \ k_i $ and make it public to the system $A$.
$$ v_1 \ = \ 7z + 8 \ , \ v_2 \ = \ 16z + 12 \ , \ v_3 \ = \ 21z + 14 $$
\item System computes
\begin{itemize}
\item $\bar{v} \ = \ \displaystyle \sum_{i=1}^{3} v_i \ \equiv 14z+4 \ \ (mod \ 15) $
\item $\bar{u} \ = \ (\bar{v})^{(\displaystyle \prod_{i=1}^{3} d_{A_i})} \ \equiv \ (\bar{v})^{5*7*3} \ \equiv \ \bar{v} \ \ (mod \ 15) $
\end{itemize}
System $A$ sends the message $M$ and signature $(r,s,t,\bar{u})$ to the system $B$ (and all the verifiers $D_i$).
\end{enumerate}


\vspace{0.3 cm}

\textbf{{\large Verification Algorithm}}
  
\vspace{0.2 cm}

Upon receiving the message $M$ and the signature $(r,s,t,\bar{u})$, each verifier $D_i$ computes $z_i \ = \ [x_{B_i}]s$ and make it public. Then, each verifier $D_i$ can separately or together verify and validate the signature by checking the following equations:
\begin{enumerate}
\item Computes $a \ = \ (\bar{u})^{\displaystyle \prod_{i=1}^{3} e_{A_i}} \ \ (mod \ n_A) $ 
\item Computes $b \ = \ t^{\displaystyle \prod_{i=1}^{2} d_{B_i}} \ \ (mod \ n_B) $
\item Check $ c \ = \ [a]P \ - \ b(y_{A_1} \ + \ y_{A_2} \ + \ y_{A_3}) \ \ (mod \ p) \ \ = \ \ r \ + \ \displaystyle \sum_{i=1}^{2} z_i \ \ (mod \ p) $
\item Lastly, signatures are accepted if $b \ = \ H_3(M,c)$.
\item If any of the condition from above two conditions fail, verifier can deny to accept the signature. In case of more than $2$ verifiers, if $50 \% $ or more verifiers deny the signature, system $B$ will return the signature to system $A$.
\end{enumerate}


\vspace{0.3 cm}

\textbf{{\large Example 2}}
  
\vspace{0.2 cm}

Let $A = \{ S_1,S_2,S_3,S_4,S_5 \} $ be the group of five signers (or provers) and $B = \{ D_1, D_2,D_3,D_4,D_5,D_6,D_7 \} $ be the group of seven designated verifiers.

Let $p \ = \ 6793 $ be a prime and $E$ be an elliptic curve over $\mathbb{Z}_p$ such that 
$$E \ = \ E_{\mathbb{Z}_p} \ = \ \{ (x,y) \in \mathbb{Z}_p \times \mathbb{Z}_p \ / \ y^2 = x^3 + 5 \}$$

Note that $|E| \ = \ 6916$ and $E$ is a cyclic group. 
Also $ n_A \ = \ 91 $ and $ n_B \ =  \ 38 $ are factors of $|E| $.

Let $P \ = \ (3245,4097)  \ , \ Q \ = \ (5223,4702) $  in  $ E $  such  that $|P| \ = \ 91 $, $|Q| \ = \ 38 $.

Let $H_3$ be a cryptographically secured hash function with arbitrary bit length.


\textbf{ Key generation algorithm for system $A$:}
\begin{enumerate}
\item System $A$ chooses prime $p_A \ = \ 7 $ and $q_A \ = \ 13 $, then computes $n_A \ = \ 7 \ * \ 13 $.
\item Each $S_i$ follow the following steps :

\begin{center}
\begin{tabular}{ |c|c|c|c|c|c| } 
 \hline
 $S_i$ & $e_{A_i}$ & $d_{A_i}$ & $x_{A_i}$ & $y_{A_i}$ \\
 \hline \hline 
 $S_1$ & $5$ & $29$ & $15$ & $15P$ \\
 $S_2$ & $11$ & $59$ & $19$ & $19P$ \\
 $S_3$ & $7$ & $31$ & $24$ & $24P$ \\
 $S_4$ & $23$ & $47$ & $18$ & $18P$ \\
 $S_5$ & $35$ & $35$ & $32$ & $32P$ \\
 \hline
\end{tabular}
\end{center}

\item public key of $S_i$ is $(e_{A_i}, y_{A_i})$.
\item private key of $S_i$ is $(d_{A_i},x_{A_i})$.

\item Note : Since, each member in system $A$ knows $\phi(n_A)$ and $e_{A_i} \ \ \forall i $, they can compute $d_{A_i}$. Thus, $d_{A_i}$ is private for members of system $B$ but not for members of system $A$.
\end{enumerate}


\textbf{ Key generation algorithm for system $B$:}
\begin{enumerate}
\item System $B$ chooses prime $p_B \ = \ 2 $ and $q_B \ = \ 19 $, then computes $ n_B \ = \ 2 \ * \ 19 $.
\item Each $D_i$ follow the following steps :

\begin{center}
\begin{tabular}{ |c|c|c|c|c|c| } 
 \hline
 $D_i$ & $e_{B_i}$ & $d_{B_i}$ & $x_{B_i}$ & $y_{B_i}$ \\
 \hline \hline 
 $D_1$ & $5$ & $11$ & $32$ & $32Q$ \\
 $D_2$ & $7$ & $13$ & $17$ & $17Q$ \\
 $D_3$ & $11$ & $5$ & $22$ & $22Q$ \\
 $D_4$ & $13$ & $7$ & $27$ & $27Q$ \\
 $D_5$ & $17$ & $17$ & $18$ & $18Q$ \\
 $D_6$ & $13$ & $7$ & $21$ & $21Q$ \\
 $D_7$ & $7$ & $13$ & $51$ & $51Q$ \\
 \hline
\end{tabular}
\end{center}

\item public key of $D_i$ is $(e_{B_i}, y_{B_i})$.
\item private key of $D_i$ is $(d_{B_i},x_{B_i})$.

\item Note : Since, each member in system $B$ knows $\phi(n_B)$ and $e_{B_i} \ \ \forall i $, they can compute $d_{B_i}$. Thus, $d_{B_i}$ is private for members of system $A$ but not for members of system $B$.
\end{enumerate}


\vspace{0.3 cm}

\textbf{{\large Signature Algorithm}}
  
\vspace{0.2 cm}

\begin{enumerate}
\item Each $S_i$ chooses a random integer $k_i$ and keep it secret and then, generates 
$$r_i \ = [{k_i}]P -[{k_i}](y_{B_1} \ + \ y_{B_2} \ + \ . \ . \ . \ + \ y_{B_e} ) \ \ (mod \ p) \ \ , \ \  s_i \ = \ [{k_i}]Q \ \ (mod  \ p)  \ \ and \  \  w_i \ = \ [{k_i}]P \ \ (mod  \ p) $$ and make it public to the system $A$ :

\begin{center}
\begin{tabular}{ |c|c|c|c|c|c| } 
 \hline
 $S_i$ & $k_i$ & $r_i$ & $s_i$ & $w_i$ \\
 \hline \hline 
 $S_1$ & $5770$ & $37P-12Q$ & $32Q$ & $37P$ \\
 $S_2$ & $2769$ & $39P-10Q$ & $33Q$ & $39P$ \\
 $S_3$ & $6476$ & $15P-6Q$ & $16Q$ & $15P$ \\
 $S_4$ & $1751$ & $22P-32Q$ & $3Q$ & $22P$ \\
 $S_5$ & $88$ & $88P-14Q$ & $12Q$ & $88P$ \\
 \hline
\end{tabular}
\end{center}

\item System computes :
\begin{itemize}
\item $r \ = \ \displaystyle \sum_{i=1}^{5} r_i  \ \equiv \ 19P-36Q \ \ (mod  \ p) $
\item $s \ = \ \displaystyle \sum_{i=1}^{5} s_i \ \equiv \ 20Q \ \ (mod  \ p) $
\item $w \ = \ \displaystyle \prod_{i=1}^{5} w_i \ \equiv \ 19P \ \ (mod  \ p) $
\item $ z \ = \ H_3(M,w) $
\item $ t \ = \ z^{(\displaystyle \prod_{i=1}^{7} e_{B_i})} \ (mod \ 38)$ 
\end{itemize} 
\item After computing $(r,s,w,z,t)$, they are made public to all $S_i$ and then each $S_i$ computes $v_i \ = \ z x_{A_i} \ + \ k_i $ and make it public to the system $A$.
$$ v_1 \ = \ 15z + 5770 \ , \ v_2 \ = \ 19z + 2769 \ , \ v_3 \ = \ 24z + 6476  \ , \ v_4 \ = \ 18z + 1751  \ , \ v_5 \ = \ 32z + 88 $$
\item System computes
\begin{itemize}
\item $\bar{v} \ = \ \displaystyle \sum_{i=1}^{5} v_i \ \equiv 17z+19 \ \ (mod \ 91) $
\item $\bar{u} \ = \ (\bar{v})^{(\displaystyle \prod_{i=1}^{5} d_{A_i})}  \ (mod \ 91) $
\end{itemize}
System $A$ sends the message $M$ and the signature $(r,s,t,\bar{u})$ to the system $B$ (and all the verifiers $D_i$).
\end{enumerate}


\vspace{0.3 cm}

\textbf{{\large Verification Algorithm}}
  
\vspace{0.2 cm}

Upon receiving the message $M$ and the signature $(r,s,t,\bar{u})$, each verifier $D_i$ computes $z_i \ = \ [x_{B_i}]s$ and make it public. Then, each $D_i$ can separately or together verify and validate the signature by checking the following equations:
\begin{enumerate}
\item Computes $a \ = \ (\bar{u})^{\displaystyle \prod_{i=1}^{5} e_{A_i}} \ \ (mod \ 91) $ 
\item Computes $b \ = \ t^{\displaystyle \prod_{i=1}^{7} d_{B_i}} \ \ (mod \ 38) $
\item Check $ c \ = \ [a]P \ - \ b(y_{A_1} \ + \ y_{A_2} \ + \ y_{A_3} \ + \ y_{A_4} \ + \ y_{A_5}) \ \ (mod \ p) \ \ = \ \ r \ + \ \displaystyle \sum_{i=1}^{7} z_i \ \ (mod \ p) $
\item Lastly, signatures are accepted if $b \ = \ H_3(M,c)$.
\item If any of the condition from above two conditions fail, verifier can deny to accept the signature. In case of more than $2$ verifiers, if $50 \% $ or more verifiers deny the signature, system $B$ will return the signature to system $A$.
\end{enumerate}

\subsection{Final Algorithm}

In this section, we have combined algorithm 1 and 2 to enhance the security of the algorithm.

\vspace{0.3 cm}

\textbf{{\large Set up}}
  
\vspace{0.2 cm}

Let $K$ be the security parameter and $sp \ = \ (G,G_\tau,\widetilde{G},g,g_A,g_B,p,e,n_A,n_B,H_1,H_2) \ \leftarrow \ Setup(K)$, where
\begin{itemize}
\item $p$ be a large prime such that $ n_A $ and $ n_B $ are factors of $p-1$,
\item $G$ is an additive cyclic group of large prime order $p$ and $g$ is a generator of $G$,
\item $G_\tau $ is a multiplicative cyclic group of prime order $p$,
\item $\widetilde{G} \ = \ Z_p ^*$ and $g_A \ , \ g_B \ \ \in \ \widetilde{G}$ such that $|g_A| \ = \ n_A $, $|g_B| \ = \ n_B $,
\item $e : \ G \times G \ \rightarrow \ G_\tau$ is symmetric admissible bilinear pairing map,
\item $H_1 : \{0,1\}^* \ \rightarrow \ G $ is a cryptographically secure hash function,
\item $H_2$ be a cryptographically secured hash function with arbitrary bit length.
\end{itemize}

\vspace{0.3 cm}

\textbf{{\large Key Generation}}
  
\vspace{0.2 cm}

\emph{\textbf{Key generation algorithm for system $A$:}}

\begin{enumerate}
\item System $A$ chooses prime $p_A$ and $q_A$, then computes $n_A \ = \ p_A \ * \ q_A$.
\item Each $S_i$ follow the following steps :
\begin{itemize}
\item choose $a_i \ \in \ G $.
\item calculate $p_i \ = \ a_i g \ (mod \ p)$.
\item randomly choose $e_{A_i}$ such that g.c.d.$(e_{A_i}, \phi(n_A)) \ = \ 1$.
\item compute $d_{A_i}$ such that $e_{A_i} \ d_{A_i} \ = \ 1 \ (mod \ \phi(n_A))$.
\item choose an integer $x_{A_i} \ \in \ Z_p ^* $.
\item calculate $y_{A_i} \ = \ g_A ^{x_{A_i}} \ (mod \ p)$.
\item public key of $S_i$ is $(p_i, e_{A_i}, y_{A_i})$.
\item private key of $S_i$ is $(a_i, d_{A_i}, x_{A_i})$.
\end{itemize}
\item System $A$, then, publishes $u \ = \ \displaystyle\sum_{i=1}^{n} p_i $.
\item Note : Since, each member in system $A$ knows $\phi(n_A)$ and $e_{A_i} \ \ \forall i $, they can compute $d_{A_i}$. Thus, $d_{A_i}$ is private for members of system $B$ but not for members of system $A$. Similarly, $p_i$ is public to system $A$ only.
\end{enumerate}


\emph{\textbf{Key generation algorithm for system $B$:}}
\begin{enumerate}
\item System $B$ chooses prime $p_B$ and $q_B$, then computes $ n_B \ = \ p_B \ * \ q_B $.
\item Each $D_i$ follow the following steps :
\begin{itemize}
\item choose $b_i \ \in \ G $.
\item calculate $q_i \ = \ b_i g \ (mod \ p) $.
\item randomly choose $e_{B_i}$ such that g.c.d.$(e_{B_i}, \phi(n_B)) \ = \ 1$.
\item compute $d_{B_i}$ such that $e_{B_i} \ d_{B_i} \ = \ 1 \ (mod \ \phi(n_B))$.
\item choose an integer $x_{B_i} \ \in \ Z_p ^* $.
\item calculate $y_{B_i} \ = \ g_B ^{x_{B_i}} \ (mod \ p)$.
\item public key of $D_i$ is $(q_i, e_{B_i}, y_{B_i})$.
\item private key of $D_i$ is $(b_i, d_{B_i}, x_{B_i})$.
\end{itemize} 
\item System $B$, then, publishes $v \ = \ \displaystyle\sum_{i=1}^{m} q_i $.

\item Note : Since, each member in system $B$ knows $\phi(n_B)$ and $e_{B_i} \ \ \forall i $, they can compute $d_{B_i}$. Thus, $d_{B_i}$ is private for members of system $A$ but not for members of system $B$.
Similarly, $q_i$ is public to system $B$ only.
\end{enumerate}


\vspace{0.3 cm}

\textbf{{\large Signature Algorithm}}
  
\vspace{0.2 cm}

System $A$ generates a signature $(\sigma,r,s,t,\bar{u})$ for a message $M$ as follows:
\begin{enumerate}
\item Each $S_i$ chooses a random integer $k_i$ and keep it secret.
\item Each $S_i$ generates $\sigma_i \ = \ e(H_1(M),a_i v)  \in \ G_\tau  $ , $r_i \ = \ g_A ^{k_i} (y_{B_1} \ y_{B_2} \ . \ . \ . \ y_{B_m} )^{- k_i} \ \ (mod \ p) $ ,  $ s_i \ = \ g_B ^{k_i} \ \ (mod  \ p) $ and $ w_i \ = \ g_A ^{k_i} \ \ (mod  \ p) $ and make it public to the system $A$.
\item System computes :
\begin{itemize}
\item $\sigma \ = \ \displaystyle\prod_{i=1}^{n} \sigma_i \ \in \ G_\tau $
\item $r \ = \ \displaystyle \prod_{i=1}^{n} r_i \ \ (mod  \ p) $
\item $s \ = \ (\displaystyle \prod_{i=1}^{n} s_i)^r \ \ (mod  \ p) $
\item $w \ = \ (\displaystyle \prod_{i=1}^{n} w_i)^r \ \ (mod  \ p) $
\item $ z \ = \ H_2(M,w) $
\item $ t \ = \ z^{\displaystyle \prod_{i=1}^{m} e_{B_i}}  \ \ (mod \ n_B)$ 
\end{itemize} 
\item After computing $(\sigma,r,s,w,z,t)$, they are made public to all $S_i$ and then each $S_i$ computes $v_i \ = \ z x_{A_i} \ + \ k_i r$ and make it public to the system $A$.
\item System computes
\begin{itemize}
\item $\bar{v} \ = \ \displaystyle \sum_{i=1}^{n} v_i \ \ (mod \ n_A) $
\item $\bar{u} \ = \ (\bar{v})^{\displaystyle \prod_{i=1}^{n} d_{A_i}} \ \ (mod \ n_A) $
\end{itemize}
System $A$ sends the message $M$ and the signature $(\sigma,r,s,t,\bar{u})$ to the system $B$ (and all the designated verifiers $D_i$).
\end{enumerate}


\vspace{0.3 cm}

\textbf{{\large Verification Algorithm}}
  
\vspace{0.2 cm}

Upon receiving the message $M$ and the signature  $(\sigma,r,s,t,\bar{u})$, each designated verifier $D_i$ computes $\zeta_i \ = \ e(H_1(M),b_i u) \ \in \ G_\tau $ and $z_i \ = \ s^{x_{B_i}} \ (mod \ p) $ and make it public. Then, each $D_i$ can separately or together verify and validate the signature by checking the following equations:
\begin{enumerate}
\item Computes $\zeta \ = \ \displaystyle\prod_{i=1}^{m} \zeta_i \ \in  \ G_\tau$.
\item Computes $a \ = \ (\bar{u})^{\displaystyle \prod_{i=1}^{n} e_{A_i}} \ \ (mod \ n_A) $ 
\item Computes $b \ = \ t^{\displaystyle \prod_{i=1}^{m} d_{B_i}} \ \ (mod \ n_B) $
\item Check $\sigma \ \equiv \ \zeta$ in $G_\tau $.
\item Check $ c \ = \ g_A ^a  \ (y_{A_1} y_{A_2} . . . y_{A_n})^{-b} \ \ (mod \ p) \ \ = \ \ r^r \displaystyle \prod_{i=1}^{m} z_i \ \ (mod \ p) $
\item Lastly, signatures are accepted if $b \ = \ H_2(M,c) $ and $\sigma \ \equiv \ \zeta$ in $G_\tau $.
\item If any of the condition from above conditions fail, verifier can deny to accept the signature. In case of more than $2$ verifiers, if $50 \% $ or more verifiers deny the signature, system $B$ will return the signature to system $A$.
\end{enumerate}

\section{Security Analysis}

Now, we show the security properties for this scheme:

\begin{theorem}
Our MSMV-SDVS scheme is correct if it runs smoothly.
\end{theorem}
\begin{proof}
The signature verification procedure is correct since:
\begin{enumerate}

\item 

$\sigma \ \ \equiv \ \displaystyle\prod_{i=1}^{n} \sigma_i \ \ \equiv \ \displaystyle\prod_{i=1}^{n}e(H_1(M),a_i v)  \ \ \equiv \ e(H_1(M),\displaystyle\sum_{i=1}^{n} a_i v) \ \ \equiv \ e(H_1(M),\displaystyle\sum_{i=1}^{n} \sum_{j=1}^{m} a_i b_j g ) $ 

$ \ \ \ \ \equiv \ e(H_1(M),\displaystyle\sum_{j=1}^{m} b_j u) \ \ \equiv \ \displaystyle\prod_{j=1}^{m} e(H_1(M), b_j u)  \ \ \equiv  \ \displaystyle\prod_{j=1}^{m} \zeta_j  \ \equiv \ \zeta $ in $G_{\tau}$

\item 

$a \ \ \equiv \ (\bar{u})^{\displaystyle\prod_{i=1}^{n} e_{A_i}} \ \ \equiv \ (\bar{v})^{(\displaystyle\prod_{i=1}^{n} d_{A_i})(\displaystyle\prod_{i=1}^{n} e_{A_i})} \ \ \equiv \ \bar{v}  \ (mod \ n_A) $

\item 

$b \ \ \equiv \ t^{\displaystyle\prod_{i=1}^{m} d_{B_i}} \ \ \equiv \ z^{(\displaystyle\prod_{i=1}^{m} e_{B_i})(\displaystyle\prod_{i=1}^{m} d_{B_i})} \ \ \equiv \ z  \ (mod \ n_B) $

\item

${g_A}^a (y_{A_1} \ . \ . \ . \ y_{A_n})^{-b} \ \ \equiv  \ {g_A}^a ({g_A}^{x_{A_1}} \ . \ . \ . \ {g_A}^{x_{A_n}})^{-b} \ \ \equiv \ {g_A}^{\bar{v}} \ ( {g_A}^{-b \displaystyle\sum_{i=1}^{n} x_{A_i} }) \ \ \equiv \ {g_A}^{r \displaystyle\sum_{i=1}^{n}k_i} \ (mod \ p) $

\item

$r^r \displaystyle\prod_{i=1}^{m} z_i \ \ \equiv \ \ r^r \displaystyle\prod_{i=1}^{m} s^{x_{B_i}} \ \ \equiv \ \ (\displaystyle\prod_{j=1}^{n} r_j)^r \ s^{\displaystyle\sum_{i=1}^{m} x_{B_i}}$ 

$\ \ \ \ \ \ \ \ \  \ \ \ \ \equiv \ ({g_A}^{\displaystyle\sum_{j=1}^{n} k_j} (y_{B_1} \ . \ . \ . \ Y_{B_m})^{-\displaystyle\sum_{j=1}^{n} k_j})^r \ ({g_B}^{r \displaystyle\sum_{j=1}^{n} k_j})^{\displaystyle\sum_{i=1}^{m} x_{B_i}} $ 

$ \ \ \ \ \ \ \ \ \ \ \ \ \ \equiv \ {g_A}^{r \displaystyle\sum_{j=1}^{n} k_j} \ (mod \ p ) $

\end{enumerate}
\end{proof}

\begin{theorem}
Our MSMV-SDVS scheme is strong designated verifier signature scheme.
\end{theorem}
\begin{proof}
We have seen that the simulated signature produced by system $B$ (or designated verifier $D$) is indistinguishable from the signature generated by system $A$. Also, the probability that the simulated signature produced randomly from the set of system $A$'s signature depends upon the randomness of $k' \ \in \ {\mathbb{Z}_{n_B} ^*}$ and it is $1/{n_B}$. Thus the two signatures have the same probability distribution, and hence the proposed scheme is a designated verifier signature scheme. 

Our scheme involves each designated verifier and their secret keys ($b_i$, $x_{B_i}$) in the verification process. Moreover, the signature ($\sigma, \ r, \ s, \ t, \ \bar{u}$) produced by system $A$ is indistinguishable from the signature ($\zeta, \ r', \ s', \ t', \ \bar{u}^{'}$) generated by system $B$. Thus, system $B$ can not convince any third party whether a signature is created by signers or verifiers. Thus, our scheme satisfies the strongness property.

\end{proof}


\begin{theorem}
Our MSMV-SDVS scheme is unforgeable unless an adversary has access to the secret keys of all signers and verifiers.
\end{theorem}
\begin{proof}
Any adversary, who does not know the private key of all signers and designated verifiers, can not forge the signature. It is computationally infeasible to forge the signature in polynomial time.

Notation : $ 1 \leq i \leq n $, $ 1 \leq j \leq m $, where $n$ and $m$ are number of signers and designated verifiers respectively.

\begin{itemize}
\item If attacker knows all $p_i$ or $q_j$, it is infeasible to compute all $a_i$ or $b_j$ in polynomial time.
\item If attacker knows all $a_i$ or $b_j$, signature $\sigma$ or $\zeta$ can be created but to compute $r,s,t,\bar{u}$, he further require $d_{A_i}, x_{A_i}, d_{B_j}, x_{B_j} \ \ \ \forall  i ,j $

\end{itemize}

Since our scheme is based on bilinear pairing, factorization problem, and discrete logarithm problem, it is computationally infeasible to compute $\sigma, r,s,t,\bar{u} $ in polynomial time.
\end{proof}


\begin{theorem}
Our MSMV-SDVS scheme is non-transferable, and protects the identity of the actual signer.
\end{theorem}
\begin{proof}
In this scheme, signature ($\zeta, \ r', \ s', \ t', \ \bar{u}^{'}$) produced by system B is indistinguishable from the signature ($\sigma, \ r, \ s, \ t, \ \bar{u}$) generated by system A. Thus, individually or collectively, Verifiers can not convince anyone that the signer generated the signature, and any third party can not determine who generated which signature. Thus our scheme is non-transferable and protects the identity of the actual signer.
\end{proof}


\section{Conclusion}

This paper presents a multi-signer strong designated multi-verifier signature scheme based on multiple cryptographic algorithms. Since all the signers generate their signatures individually and make them public to the system to generate a common signature, this scheme is non-interactive. Also, our scheme requires the authorization of each participating signer, making it applicable to various dimensions using blockchain technology. Our scheme satisfies strongness, unforgeability, and non-transferability properties in the random oracle model.

\newpage

 \bibliographystyle{plain}
 \bibliography{References_paper1}

\end{document}